\documentclass{article}
\usepackage{amsmath,amsthm,amsfonts,amssymb}
\usepackage{url}
\usepackage{xspace}
\usepackage[names]{xcolor}
\usepackage[noend,noline,ruled]{algorithm2e}
\usepackage{graphicx}
\usepackage[margin=1.25in]{geometry}

\newtheorem{theorem}{Theorem}
\newtheorem{claim}[theorem]{Claim}
\newtheorem{corollary}[theorem]{Corollary}
\newtheorem{lemma}[theorem]{Lemma}
\newtheorem{proposition}[theorem]{Proposition}
\newtheorem{conjecture}{Conjecture}
\newtheorem{oproblem}{Open Problem}
\theoremstyle{definition}
\newtheorem{definition}{Definition}
\newtheorem*{remark}{Remark}
\theoremstyle{plain}

\newcommand{\xor}{\textsc{xor}}
\newcommand{\maj}{\textsc{maj}}
\DeclareMathOperator*{\E}{\mathrm{E}}
\newcommand{\eps}{\varepsilon}
\DeclareMathOperator{\ind}{\mathbf{1}}
\DeclareMathOperator{\deq}{:=}
\newcommand{\bit}{\{0,1\}}
\newcommand{\bitsn}{\bit^n}
\newcommand{\bitsm}{\bit^m}
\newcommand{\bitz}{\{0,1,*\}}

\newcommand{\N}{{N}}

\newcommand{\cA}{\mathcal{A}}
\newcommand{\cB}{\mathcal{B}}
\newcommand{\cX}{\mathcal{X}}

\newcommand{\errordep}{(\frac{n}{\log n})^{1/3}}

\newcommand{\bluered}{\textsc{BlueRed}\xspace}
\newcommand{\gapID}{\textsc{GapID}\xspace}
\newcommand{\gapid}{\gapID}
\newcommand{\PtrFcn}{\textsc{PtrFcn}\xspace}
\newcommand{\EncFcn}{\textsc{EncFcn}\xspace}

\newcommand{\black}{\textsc{black}\xspace}
\newcommand{\blue}{\textsc{blue}\xspace}
\newcommand{\red}{\textsc{red}\xspace}

\newcommand{\valx}{\textsc{value}\xspace}
\newcommand{\val}{\textsc{value}\xspace}
\newcommand{\row}{\textsc{row}\xspace}
\newcommand{\col}{\textsc{col}\xspace}

\newcommand{\goos}{G{\"o\"o}s\xspace}

\newcommand{\ddt}{\mathrm{D}}
\newcommand{\rdt}{\mathrm{R}}
\newcommand{\rcc}{\mathrm{R}^{\mathrm{cc}}}
\newcommand{\rccow}{\mathrm{R}^{\mathrm{cc},\to}}
\newcommand{\rccpriv}{\mathrm{R}^{\mathrm{cc,priv}}}
\newcommand{\rdtde}{\rdt_{\delta,\eps}}

\title{Optimal Separation and Strong Direct Sum for Randomized Query Complexity}

\author{Eric Blais \\ University of Waterloo \\ \texttt{eric.blais@uwaterloo.ca}
   \and Joshua Brody \\ Swarthmore College \\ \texttt{brody@cs.swarthmore.edu}}

\begin{document}

\maketitle

\begin{abstract}
We establish two results regarding the query complexity of bounded-error randomized algorithms.
\begin{description}
\item[Bounded-error separation theorem.] There exists a total function
  $f : \{0,1\}^n \to \{0,1\}$ whose $\epsilon$-error randomized query
  complexity satisfies
  $\overline{\mathrm{R}}_\eps(f) = \Omega( \mathrm{R}(f) \cdot \log
  \frac1\epsilon)$.
\item[Strong direct sum theorem.] For every function $f$ and every
  $k \ge 2$, the randomized query complexity of computing $k$
  instances of $f$ simultaneously satisfies
  $\overline{\mathrm{R}}_\eps(f^k) = \Theta(k \cdot
  \overline{\mathrm{R}}_{\frac\epsilon k}(f))$.
\end{description}
As a consequence of our two main results, we obtain an optimal superlinear
direct-sum-type theorem for randomized query complexity: there exists a
function $f$ for which $\mathrm{R}(f^k) = \Theta( k \log k \cdot \mathrm{R}(f))$.
This answers an open question of Drucker (2012). Combining this result
with the query-to-communication complexity lifting theorem of
G\"o\"os, Pitassi, and Watson (2017), this also shows that there is a
total function whose public-coin randomized communication complexity
satisfies $\mathrm{R}^{\mathrm{cc}}(f^k) = \Theta( k \log k \cdot \mathrm{R}^{\mathrm{cc}}(f))$, answering a
question of Feder, Kushilevitz, Naor, and Nisan (1995).
\end{abstract}

\section{Introduction}

We consider two fundamental questions related to the query complexity
of functions in the bounded-error randomized setting: how the
randomized query complexity of total functions scales with the
allowable error $\epsilon$ (the \emph{separation} problem), and how
the query complexity of computing $k$ instances of a function scales
with the complexity of computing only 1 instance of the same function
(the \emph{direct sum} problem). Standard folklore arguments give
upper bounds on how much the randomized query complexity can depend on
$\epsilon$ and on $k$ in these two problems; the results described
below show that these well-known upper bounds are tight in general.

A randomized algorithm $\cA$ \emph{computes} a function
$f : \cX^n \to \bit$ over a finite set $\cX^n$ \emph{with error
  $\epsilon \ge 0$} if for every input $x \in \cX^n$, the algorithm
outputs the value $f(x)$ with probability at least $1-\epsilon$.  The
\emph{query cost} of $\cA$ is the maximum number of coordinates of $x$
that it queries, with the maximum taken over both the choice of input
$x$ and the internal randomness of $\cA$. The \emph{$\epsilon$-error
  (worst-case) randomized query complexity} of $f$ (also known as the
\emph{randomized decision tree complexity} of $f$) is the minimum
query complexity of an algorithm $\cA$ that computes $f$ with error at
most $\epsilon$. We denote this complexity by $\rdt_\epsilon(f)$, and
we write $\rdt(f) := \rdt_{\frac13}(f)$ to denote the $\frac13$-error
randomized query complexity of $f$.

Another natural measure for the query cost of a randomized algorithm
$\cA$ is the \emph{expected} number of coordinates of an input $x$
that it queries. Taking the maximum expected number of coordinates
queried by $\cA$ over all inputs yields the \emph{average query cost}
of $\cA$. The minimum average query complexity of an algorithm $\cA$
that computes a function $f$ with error at most $\epsilon$ is the
\emph{average $\epsilon$-error query complexity} of $f$, which we
denote by $\overline{\rdt}_{\epsilon}(f)$. We again write
$\overline{\rdt}(f) := \overline{\rdt}_{\frac 13}(f)$. Note that
$\overline{\rdt}_0(f)$ corresponds to the standard notion of
\emph{zero-error randomized query complexity} of $f$.

\subsection{Our Results}

\paragraph*{Bounded-Error Separation Theorem for Query Complexity }
One of the first tricks that one learns in the study of randomized
algorithm is \emph{success amplification}: it is possible to cheaply
reduce the error of a randomized algorithm from $\frac13$ to any
$\eps > 0$ by running the algorithm $O(\log \frac1\eps )$ times and
outputting the most frequent answer.  In the context of randomized
query complexity, this means that for every function
$f : \bitsn \rightarrow \bit$,
\begin{equation}
\label{eqn:separation-ub}
\rdt_\eps(f) = O\big(\rdt(f)\cdot \log \tfrac1\eps \big).
\end{equation}
When considering partial functions, it is easy to see that the success
amplification trick is optimal, as there are partial functions for
which this relationship is tight (see Section~\ref{sec:GapID}).
However, in the case of total functions, for many natural functions
such as the majority function, parity function, dictator function,
etc., the stronger bound $\rdt_\epsilon(f) = O\big( \rdt(f) \big)$
holds and until now it was not known whether the bound
in~\eqref{eqn:separation-ub} is tight for \emph{any} total function.
In fact, even separations between zero-error and $\frac13$-error
randomized query complexity were not known until very recently, when
Ambainis et al.~\cite{AmbainisBBLSS17} showed that there exists a
total function $f$ for which
$\overline{\rdt}_0(f) = \widetilde{\Omega}(\rdt(f)^2)$. Similarly,
other separations between randomized query complexity and other
measures of complexity have also only been established very
recently~\cite{MukhopadhyayS15,AaronsonBK16,AmbainisKK16,AnshuBBGJKLS16,AmbainisBBLSS17}.

In this work, we give the first separation within the bounded-error
randomized query complexity setting. Our separation shows that the
bound in~\eqref{eqn:separation-ub} is optimal in general.

\begin{theorem}
\label{thm:separation}
For infinitely many values of $n$ and every
$2^{-\errordep} < \eps \le \frac13$, there exists a total function
$f:\bitsn \rightarrow \bit$ with randomized query complexity
\[
%\rdt_\eps(f) \ge
\overline{\rdt}_\eps(f) = \Omega\big(\rdt(f) \cdot \log\tfrac{1}{\eps} \big).
\]
\end{theorem}

Note that by the trivial relation
$\overline{\rdt}_\eps(f) \le \rdt_\eps(f)$ between average and
worst-case randomized query complexity, Theorem~\ref{thm:separation}
implies the existence of a function $f$ for which
$\rdt_\eps(f) \ge \Omega\big(\rdt(f) \cdot \log\tfrac{1}{\eps} \big)$
and
$\overline{\rdt}_\eps(f) \ge \Omega\big(\overline{\rdt}(f) \cdot
\log\tfrac{1}{\eps} \big)$, giving optimal separations in both the
worst-case randomized query complexity and average query complexity
settings.

\paragraph*{Strong Direct Sum Theorem}
The \emph{direct sum problem} asks how the cost of computing a
function $f$ scales with the number $k$ of instances of the function
that we need to compute. This problem has received a considerable
amount of attention in the context of query
complexity~\cite{ImpagliazzoRW94,BenAsherN95,NisanRS99,Shaltiel03,JainKS10,BenDavidK18,Drucker12},
communication
complexity~\cite{KarchmerRW95,FederKNN95,ChakrabartiSWY01,BarYossefJKS04,MolinaroWY13,BarakBCR13},
and beyond.

Given a function $f : \bitsn \to \bit$ and a parameter $k \ge 2$,
define $f^k : \bit^{n \cdot k} \to \bit^k$ by setting
$f^k(x^{(1)},\ldots,x^{(k)}) = \big(
f(x^{(1)}),\ldots,f(x^{(k)})\big)$. A simple union bound argument
shows that the randomized query complexity of $f^k$ is bounded above
by
\begin{equation}
\label{eq:directsum-ub}
R_\epsilon(f^k) = O\big( k \cdot R_{\frac \epsilon k}(f) \big)
\end{equation}
since we can run a randomized algorithm $\cA$ that computes $f$ with
error at most $\frac \epsilon k$ on each of the $k$ instances. An
analogous upper bound holds in the average query complexity setting as
well.

Jain, Klauck, and Santha~\cite{JainKS10} first considered the problem
of showing a direct sum theorem for randomized query complexity. They
showed that for every function $f$ and for small enough constant
$\delta > 0$,
$ \rdt_\epsilon(f^k) \ge \delta^2 k \cdot
\rdt_{\frac{\epsilon}{1-\delta} + \delta}(f).  $ Note that in this
inequality, the allowable error on the right-hand side of the equation
is \emph{larger} than the $\epsilon$ error parameter, in contrast to
the upper bound where it is (much) smaller. Ben-David and
Kothari~\cite{BenDavidK18} obtained an improved direct sum theorem
holds, showing that
$ \overline{\rdt}_\epsilon(f^k) \ge k \cdot
\overline{\rdt}_{\epsilon}(f) $ holds for every function. This result
is formally stronger since it implies the Jain--Klauck--Santha bound,
but it also does not show that the error parameter on the
right-hand-side of the inequality needs to be smaller than $\epsilon$,
as it is in the upper bound~\eqref{eq:directsum-ub}.

We show that the bound in~\eqref{eq:directsum-ub} is tight in the
average-case query complexity model.

\begin{theorem}
\label{thm:directsum}
For every function $f:\bitsn \rightarrow \bit$, every $k \ge 2$, and
every $0 \le \epsilon \le \frac1{20}$,
\[
\overline{\rdt}_\epsilon(f^k) = \Omega\left( k \cdot \overline{\rdt}_{\frac{\epsilon} k}(f) \right).
\]
\end{theorem}

We establish Theorem~\ref{thm:directsum} by proving a corresponding
strong direct sum theorem in the distributional setting, as we discuss
in more details in Section~\ref{sec:proof-overviews}. It remains open
to determine whether a similar strong direct sum theorem holds in the
worst-case randomized query complexity model. However, in that setting
Shaltiel~\cite{Shaltiel03} has shown that a proof of such a direct sum
theorem \emph{can't} be obtained via a corresponding theorem in the
distributional setting, as a counterexample shows that direct sum
theorems do not hold in this setting in general.

\subsection{Applications}

\paragraph*{Superlinear Direct-Sum-Type Theorem for Query Complexity}
Combining~\eqref{eqn:separation-ub} and~\eqref{eq:directsum-ub}, we
obtain a bound on the cost of computing $k$ instances of a function
$f$ with bounded (constant) error and the cost of computing a single
instance of the same function:
\begin{equation}
\label{eq:ds-ub}
\rdt(f^k) = O\big( k \log k \cdot \rdt(f) \big).
\end{equation}
Drucker~\cite[Open problem 2]{Drucker12} asked if the superlinear
dependence on $k$ in~\eqref{eq:ds-ub} is necessary for any total
function $f$. Theorems~\ref{thm:separation} and~\ref{thm:directsum}
give a positive answer to this question.

\begin{corollary}
\label{cor:directsum-rdt}
There exists a total function $f : \{0,1\}^n \to \{0,1\}$ such that
for all $1 \leq k \leq 2^{\errordep}$,
\[
\rdt(f^k) = \Theta\big(k\log k \cdot \rdt(f)\big).
\]
\end{corollary}

Note that Corollary~\ref{cor:directsum-rdt} stands in contrast to the
quantum query complexity setting, where such a superlinear dependence
on $k$ is not required~\cite{BuhrmanNRW07}.

\paragraph*{Superlinear Direct-Sum-Type Theorem for Communication Complexity}
Let $\rcc(f)$ denote the minimum amount of communication required of a
public-coin randomized protocol that computes a function
$f : \bitsm \times \bitsn \to \{0,1\}$ with error at most
$\frac13$. As in the query complexity model, the communication
complexity of the function $f^k$ is bounded above by
\begin{equation}
\label{eq:dscc-ub}
\rcc(f^k) = O\big( k \log k \cdot \rcc(f) \big).
\end{equation}
Feder, Kushilevitz, Naor, and Nisan~\cite{FederKNN95} showed that this
upper bound is not tight in general, as the equality function
satisfies
$\rcc(\textsc{Eq}^k) = O\big(k \cdot \rcc(\textsc{Eq})
\big)$.\footnote{In fact, Feder et al.~showed that the
  \emph{private-coin} randomized communication complexity of
  \textsc{Eq} satisfies the stronger relation
  $\rccpriv(\textsc{Eq}^k) = o\big(k \cdot \rccpriv(\textsc{Eq})
  \big)$; their construction also directly establishes the result
  stated in the main text.}  They then asked whether
$\rcc(f^k) = O(k \cdot \rcc(f))$ holds for all functions or
not~\cite[Open problem 2 in \S7]{FederKNN95}.

In the last few years, there has been much work on related direct sum questions. Molinaro, Woodruff, and Yaroslavtsev~\cite{MolinaroWY13,MolinaroWY15} showed that in the \emph{one-way} communication complexity model, the equality function does satisfy the superlinear direct sum bound
$\rccow(\textsc{Eq}^k) = \Theta\big(k \log k \cdot \rccow(\textsc{Eq})
\big)$.
In the two-way communication complexity model that we consider,
Barak, Braverman, Chen, and Rao~\cite{BarakBCR13} showed that every function $f$ satisfies the direct sum $R(f^k) = \widetilde\Omega\big(\sqrt{k}\,R(f)\big)$, and this bound remains the state of the art as far as we know.
Using the connection between information complexity and amortized communication complexity of Braverman and Rao~\cite{BravermanR14}, Ganor, Kol, and Raz~\cite{GanorKR14} also showed that there is a partial function whose distributional communication complexity is exponentially larger than its amortized distributional communication complexity, showing that a tight direct sum theorem cannot hold in general in this setting.
None of these results, however, answer Feder et al.'s original question.

Corollary~\ref{cor:directsum-rdt} combined with the randomized
query-to-communication lifting theorem of \goos, Pitassi, and
Watson~\cite{GoosPW17} answers Feder et al.'s question by showing that
there is a function $f$ for which the bound in~\eqref{eq:dscc-ub} is
tight.

\begin{corollary}
\label{cor:directsum-rcc}
There is a constant $c > 0$ and a total function
$f : \{0,1\}^n \times \{0,1\}^n \to \{0,1\}$ such that for all
$1 \leq k \leq 2^{n^{c}}$,
\[
\rcc(f^k) = \Theta\big(k\log k \cdot \rcc(f)\big).
\]
\end{corollary}

\subsection{Proof Overviews}
\label{sec:proof-overviews}

\paragraph*{Bounded-Error Separation Theorem}
The proof of Theorem~\ref{thm:separation} is established by following
the general approach used to great effect by Ambainis et
al.~\cite{AmbainisBBLSS17}: first, identify a partial function $f$ for
which the query complexity separation holds, then design a variant of
the \goos--Pitassi--Watson (GPW) pointer function~\cite{GoosPW15} that
``embeds'' the partial function into a total function and preserves
the same separation.

The first step in this plan is accomplished by observing that the
partial \emph{gap identity} function
$\gapid : \bitsm \rightarrow \bitz$ defined by
\[
\gapid(x) = \begin{cases} 1 & \text{if } |x| = 0,\\ 0 & \text{if } |x| = \lfloor\frac{m}{2}\rfloor,\\ * & \text{otherwise}\end{cases}
\]
satisfies
$\overline{\rdt}_\eps(\gapid) = \Theta\big(\rdt(\gapid) \cdot
\log\tfrac{1}{\eps} \big)$ for every $\epsilon \ge 2^{-m}$.

Ambainis et al.~\cite{AmbainisBBLSS17} also used (essentially) the
same gap identity function to establish the separation
$\overline{\rdt}_0(f) = \widetilde{\Omega}\big(\rdt(f)^2\big)$. In
constructing a GPW pointer function analogue of the $\gapid$ function,
however, Ambainis et al.~lose a few logarithmic factors: their
construction shows that there exists a total function
$f : \bitsn \to \bit$ with $\epsilon$-error randomized query
complexity that satisfies
$\rdt_\epsilon(f) = O(\sqrt{n} \log^2n \log \frac1\eps)$ and
$\rdt_\epsilon(f) = \Omega(\sqrt{n} \log \frac1\epsilon)$. The
polylogarithmic gap between those two bounds is not particularly
important when comparing this query complexity to the zero-error
randomized query complexity
$\overline{\rdt}_0(f) = \widetilde{\Omega}(n)$ of the same function,
but it makes it impossible to obtain any separation at all between
$\rdt(f)$ and $\rdt_\eps(f)$ whenever
$\epsilon = \Omega(n^{-\log n})$. To prove
Theorem~\ref{thm:separation}, we need a new variant of the GPW pointer
function whose analysis avoids \emph{any} gap that is a non-constant
function of $n$.

At a high-level, GPW pointer functions are constructed by defining an
$n \times m$ array of cells, whose values are taken from some
(typically fairly large) alphabet $\Sigma$. The first logarithmic gap
in Ambainis et al.'s upper and lower bounds occurs because the upper
bound is measured in terms of the number of \emph{bits} queried by the
algorithm while the lower bound is in terms of the number of
\emph{cells} queried by an algorithm. To eliminate this gap, we must
either reduce the size of the alphabet from $|\Sigma| = O(\log n)$ to
a constant size or modify the analysis so that both the upper and
lower bounds are in terms of bit-query complexity. We do the latter,
using the notion of \emph{resilient functions}~\cite{ChorGHFRS85} to
show that an algorithm must query a constant fraction of the bits of a
cell to learn \emph{anything} about the contents of that cell.
Resilient functions were introduced by Chor et al.~\cite{ChorGHFRS85},
who gave an essentially optimal construction using basic linear
algebra and the probabilistic method.  Sherstov recently created a
gadget~\cite{Sherstov18} resilient to approximate polynomial
degree.  This gadget is both similar in construction
to~\cite{ChorGHFRS85} and in motivation to our work; it too removes
some loss due to function inputs coming from large alphabets.

The second logarithmic gap in Ambainis et al.'s construction occurs
because the location of the ``special'' cells that an algorithm seeks
to discover in the GPW pointer function can be found by following a
binary tree structure; the upper bound accounts for the $\log n$ cell
queries an algorithm requires to follow this structure while the lower
bound holds even if an algorithm finds these special cells in a single
query. We bypass this problematic issue with a simple but powerful
observation: in our setting, once we use resilient functions to encode
the contents of each cell, there is no longer any requirement to keep
the size $|\Sigma|$ of the alphabet for each cell in the GPW pointer
function to be polylogarithmic in $n$ and so we can include a
\emph{lot} more information in each cell without affecting the query
complexity gap. We use this flexibility to replace pointers to the
root of a binary tree structure with direct pointers to all the
special cells in its leaves.

The details of the proof of Theorem~\ref{thm:separation} are presented
in Section~\ref{sec:separation}.

\paragraph*{Strong Direct Sum Theorem}
Our proof of the strong direct sum theorem proceeds by establishing an
analogous result in the setting of distributional query
complexity. The \emph{$\epsilon$-error distributional complexity} of
$f : \bitsn \to \bit$ with respect to the distribution $\mu$ on
$\{0,1\}^n$, denoted by $\ddt^\mu_{\epsilon}(f)$, is the minimum query
complexity of a deterministic algorithm that computes the value $f(x)$
correctly with probability at least $1-\epsilon$ when $x$ is drawn
from $\mu$.

The distributional complexity approach is also the one used in prior
work on direct sum theorems for query
complexity~\cite{JainKS10,BenDavidK18}. The challenge with this
approach, however, is that a strong direct sum theorem for
distributional query complexity does \emph{not} hold in general, as
Shaltiel~\cite{Shaltiel03} demonstrated (see also \S4
in~\cite{Drucker12}): there exists a function $f$ and a distribution
$\mu$ on $f$'s domain for which
$\ddt^{\mu^k}_\epsilon(f^k) = O\big(\epsilon k
\ddt^\mu_\epsilon(f)\big)$.

A similar barrier to strong direct sum theorems exists in the
communication complexity setting. Molinaro, Woodruff, and
Yaroslavstev~\cite{MolinaroWY13,MolinaroWY15} bypassed this barrier by
considering randomized protocols that are allowed to abort with some
bounded probability. They were then able to show that the information
complexity of such communication protocols (in both the one-way and
two-way communication models) satisfies a strong direct sum property.

Following an analogous approach, we consider randomized algorithms
that are allowed to \emph{abort} (or, equivalently, to output some
value $\bot$ that corresponds to ``don't know'') with some probability
at most $\delta$. The \emph{$\epsilon$-error, $\delta$-abort
  randomized query complexity} of a function $f$ is denoted by
$R_{\delta,\epsilon}(f)$. With a natural extension of Yao's minimax
principle, we can obtain bounds on this randomized query complexity by
considering the corresponding \emph{$\epsilon$-error, $\delta$-abort
  distributional complexity} $\ddt^\mu_{\delta,\epsilon}(f)$ of a
function $f$, which is the minimum query complexity of deterministic
algorithms must err with probability at most $\epsilon$ and abort with
probability at most $\delta$ when inputs are drawn from the
distribution $\mu$. We show that a strong direct sum theorem does hold
in this setting.

\begin{lemma}
\label{lem:main-ds}
There exists a constant $c$ such that for every function
$f : \{0,1\}^n \to \{0,1\}$, every distribution $\mu$ on $\{0,1\}^n$,
and every $0 \le \delta, \epsilon \le \frac1{40}$,
\[
  \ddt^{\mu^k}_{\delta,\epsilon}(f^k) = \Omega\left( k \cdot
    \ddt^\mu_{\frac1{3},\frac{c\cdot \epsilon}{k}}(f) \right).
\]
\end{lemma}

The proof of Theorem~\ref{thm:directsum} is then obtained from this lemma by showing that an analogue of Yao's minimax principle holds for algorithms that can both err and abort. The full details of the proofs of Lemma~\ref{lem:main-ds} and
Theorem~\ref{thm:directsum} are presented in
Section~\ref{sec:ds}.

\section{Bounded-Error Separation Theorem}
\label{sec:separation}

We complete the proof of Theorem~\ref{thm:separation} in this section.
In Section~\ref{sec:ptrfcn}, we first define the pointer function
$\PtrFcn$ at the heart of the proof.  In
Sections~\ref{sec:GapID}--\ref{sec:ptrfcn-lb}, we establish a lower
bound on the query complexity of the $\PtrFcn$ function via reductions
from the $\gapID$ function, and in Section~\ref{sec:ptrfcn-ub}, we
provide a matching upper bound on this query complexity.  We complete
the proof of Theorem~\ref{thm:separation} in
Section~\ref{sec:thmsep-proof} by combining these results with the use
of resilient functions.

\subsection{Pointer Function}
\label{sec:ptrfcn}

The total function at the heart of the proof of
Theorem~\ref{thm:separation} is a variant of the
\goos--Pitassi--Watson pointer function $\PtrFcn$ that we define
below.  Let $[n]$ denote the set $\{1,\ldots, n\}$.

\begin{figure}
  \centering
  \includegraphics[height=6cm]{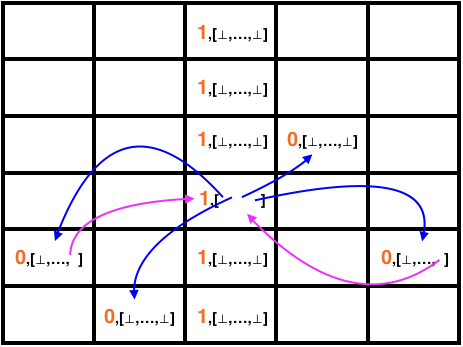}
  \caption{A $1$-input for $\PtrFcn$.  $\PtrFcn(x) =1$ iff there is a
    unique column whose cells all have value 1; there is a special
    cell within this column which has nontrivial row pointers; these
    pointers all point to cells with value 0; and half of these linked
    cells point back to the special cell.  Note: blank cells in the
    figure represent inputs that can be arbitrary.}
\end{figure}

Define
$\Gamma = \{0,1\} \times ([n] \cup \{\bot\})^m \times ([m] \cup \{\bot\})
$ to be the set of symbols $\sigma$ that encode a \emph{value} that
we denote by $\val(\sigma)$, $m$ \emph{row
pointers} that we denote by $\row_1(\sigma),\ldots,\row_m(\sigma)$, and
one \emph{column pointer} that we denote $\col(\sigma)$.

The function $\PtrFcn : \Gamma^{n \times m} \to \{0,1\}$ is defined as
follows. First, we represent an input $x \in \Gamma^{n \times m}$ as
an $n \times m$ grid of \emph{cells}. We say that a column
$j^* \in [m]$ is \emph{special} for $x$ if $\val(x_{i,j^*}) = 1$ for
every $1 \le i \le n$. Then $\PtrFcn(x) = 1$ if and only if
\begin{itemize}
\item There is a unique column $j^*$ that is special for $x$;
\item Within the special column $j^*$, there is a unique cell $i^*$
  called the \emph{special cell};
\item $\row_j(x_{i,j^*}) = \bot$ for all $i \neq i^*$ and all
  $j \neq j^*$;
\item For all $j \neq j^*$, let $i_j \deq \row_j(x_{i^*,j^*})$.  Then,
  we have
  \begin{itemize}
  \item $\val(x_{i_j,j}) = 0$ (i.e., all cells pointed to by the
    special cell have value 0)
  \item
    $|\{j \neq j^* : \col(x_{i_j,j}) = j^* \wedge
    \row_{j^*}(x_{i_j,j}) = i^*\}| = \lfloor \frac{m-1}{2} \rfloor$
    (i.e., \emph{half} the cells pointed to by the special cell point
    back to the special cell)
  \end{itemize}
\end{itemize}

We call the cells $(i_j,j)$ \emph{linked cells}; linked cells that
point back to the special cell are \emph{good}.  In summary,
$\PtrFcn(x) = 1$ if (i) there is a special column, (ii) within the
special column, there is a special cell, (iii) all cells in the
special column that are not the special cell have
$\row_j(x_{i,j*}) = \bot$ for all $j \neq j^*$, (iv) each linked cell
has value $0$, and (v) exactly half of the linked cells are good.

The following simple claim will be useful in obtaining the query
complexity lower bound for $\PtrFcn$.

\begin{claim}
\label{claim:specialcell}
Let $\cA$ be an $\eps$-error randomized algorithm for $\PtrFcn$.  Let
$z \in \PtrFcn^{-1}(1)$, and let $(i^*,j^*)$ be the special cell of
$z$.  Then $\cA(z)$ probes $(i^*,j^*)$ with probability at least
$1-2\eps$.
\end{claim}
\begin{proof}
  Let $\bar{z}$ be the same input as $z$ except that
  $\valx(i^*,j^*) = 0$.  Then $\PtrFcn(z)\neq \PtrFcn(\bar{z})$ but
  $z,\bar{z}$ differ only on the special cell.  Whenever $\cA$ doesn't
  probe the special cell, it must output the same value for $z$ and
  $\bar{z}$, so it errs on either $z$ or $\bar{z}$.  By the error
  guarantee of $\cA$ and a union bound, the probability that $\cA$
  doesn't probe cell $(i^*,j^*)$ is at most $2\eps$.
\end{proof}

\subsection{Lower Bound on the Query Complexity of GapID}
\label{sec:GapID}

We begin the proof of Theorem~\ref{thm:separation} by establishing a
(simple, asymptotically optimal) lower bound on the average query
complexity of the $\gapID$ function.

\begin{lemma}
\label{lem:gapid-lb}
For every $m \ge 2$ and every $\epsilon < \frac12$,
$ \overline{\rdt}_\epsilon(\gapID) = \Omega\left( \min\{\log
  \tfrac1\epsilon, m\} \right).  $
\end{lemma}

\begin{proof}
  Fix any $\epsilon \ge 2^{-\frac23 m}$. We will show that
  $\overline{\rdt}_\epsilon(\gapID) = \Omega\left( \log
    \tfrac1\epsilon \right)$.  This suffices to complete the proof of
  the theorem since it implies that for any
  $\epsilon < 2^{-\frac23 m}$,
  $\overline{\rdt}_\epsilon(\gapID) \ge \overline{\rdt}_{2^{-\frac23
        m}}(\gapID) = \Omega(m)$.

  Let $\cA$ be a randomized algorithm that computes $\gapID$ with
  error probability at most $\epsilon$. Let $Q \subseteq [m]$ be a
  random variable that denotes the set of coordinates queried by
  $\cA$, and let $\xi := \xi(Q,x)$ denote the event that each
  coordinate of the input $x$ queried by the algorithm has the value
  $0$. Note that when the event $\xi(Q,x)$ occurs, $\cA$ has the same
  behavior on input $x$ as it does on the input $0^m$. Since
  $\gapID(0^m) = 1$ and $\cA$ has error probability at most
  $\epsilon$, this means that for every input $x \in \{0,1\}^m$,
\[
\Pr[ \cA(x) = 0 \wedge \xi] = \Pr[ \cA(0^m) = 0 \wedge \xi] \le
\Pr[ \cA(0^m) = 0] \le \epsilon
\]
and so $\Pr[ \cA(x) = 1] \ge \Pr[ \cA(x) = 1 \wedge \xi] \ge \Pr[\xi] - \epsilon$.

Define $\mu$ to be the uniform distribution on all inputs
$x \in \{0,1\}^m$ with $|x| = m/2$. To err with probability at most
$\epsilon$ on those inputs, the algorithm $\cA$ must satisfy
$\Pr\left[ \cA(x) = 1 \right] \le \epsilon$ for every $x$ in the
support of $\mu$.  Combining this upper bound with the previous lower
bound, we therefore have that
\begin{equation}
\label{eq:xi-ub}
\Pr_{x \sim \mu,\, Q}[ \xi ] - \epsilon \le \E_{x \sim \mu}\big[ \Pr[ \cA(x) = 1]\big] \le \epsilon
\qquad \Longrightarrow \qquad
\Pr_{x \sim \mu,\, Q} \left[  \xi \right] \le 2\epsilon.
\end{equation}
For any value $1 \le q \le \frac m3$,
\begin{align*}
\Pr_{x \sim \mu,\, Q} \big[ \xi \;\big|\; |Q| = q\big]
= \frac{\binom{m-q}{m/2}}{\binom{m}{m/2}}
&= \frac{\frac m2 (\frac
    m2-1) \cdots (\frac m2 - q + 1)}{m (m-1) \cdots (m-q+1)} \\
&> \left(\frac{\frac m2 - q}{m-q}\right)^q > \left( \frac12 - \frac{q}{2(m-q)} \right)^q \ge 4^{-q}.
\end{align*}
Therefore,
\[
  \Pr_{x \sim \mu,\, Q} \big[ \xi \;\big|\; |Q| \le \tfrac12 \log
  \tfrac1{4\epsilon}\big] > 4^{-\frac12 \log\frac1{4\epsilon}} =
  4\epsilon.
\]
Combining this inequality with~\eqref{eq:xi-ub}, we obtain
\[
  2\epsilon \ge \Pr_{x \sim \mu,\, Q} \left[ \xi \right] \ge \Pr\left[
    |Q| \le \tfrac12 \log \tfrac1{4\epsilon}\right] \cdot \Pr_{x \sim
    \mu,\, Q} \left[ \xi \;\middle|\; |Q| \le \tfrac12 \log
    \tfrac1{4\epsilon}\right] > \Pr\left[ |Q| \le \tfrac12 \log
    \tfrac1{4\epsilon}\right] \cdot 4\epsilon.
\]
Rearranging the inequality yields
$\Pr\left[ |Q| \le \tfrac12 \log \tfrac1{4\epsilon}\right] < \frac12$
and so the average query complexity of $\cA$ is bounded below by
\[
  \E\big[|Q|\big] > \tfrac12 \log\tfrac1{4\epsilon} \cdot \Pr\left[
    |Q| > \tfrac12 \log \tfrac1{4\epsilon}\right] > \tfrac14
  \log\tfrac1{4\epsilon}. \qedhere
\]
\end{proof}

\subsection{Lower Bound on the Query Complexity of BlueRed}

We wish to relate the average query complexity of $\PtrFcn$ to that of
the $\gapID$ function. We do this by relating both query complexities
to that of another partial function that we call $\bluered$.

Let $\Sigma \deq \{\black,\blue, \red\}$, and call a symbol
\emph{colored} if it is not \black.  The input is an $n\times m$ grid
of entries from $\Sigma$, with the promise that each column contains a
unique colored entry, and either all colored entries are \red, or half
the colored entries are \blue.  Formally, we define
$\bluered : \Sigma^{n \times m} \rightarrow \bitz$ as follows:
\[
\bluered(x) = \begin{cases}
1 & \text{if each column has $1$ colored entry \& all colored entries are \red}, \\
0 & \text{if each column has $1$ colored entry \& exactly $\lfloor\frac{m}{2}\rfloor$ entries are \blue}, \\
* & \text{otherwise.}
\end{cases}
\]
The following reduction shows that the average query complexity of $\bluered$ is $\Theta(n)$ times as large as that of the $\gapID$ function.

\begin{lemma}
\label{lem:bluered-lb}
For every $\epsilon > 0$,
$
\overline{\rdt}_{\eps}(\bluered) \ge \frac{n}{4}\cdot\overline{\rdt}_{\eps}(\gapID).
$
\end{lemma}

\begin{proof}
  Fix any algorithm $\cA$ that computes $\bluered$ with error at most
  $\epsilon$ and has expected query cost
  $c = \overline{\rdt}_{\eps}(\bluered)$.  We will use $\cA$ to
  construct an algorithm $\cB$ that computes $\gapID$ with error at
  most $\epsilon$ and expected cost $4c/n$.

Given an input $x \in \{0,1\}^m$, the algorithm $\cB$ constructs an instance of
the $\bluered$ problem in the following way.
First, it generates indices $i_1,\ldots,i_m \in [n]$ independently and uniformly at random.  Then it defines
\[
y_{i,j} = \begin{cases}
\red & \text{if } i = i_j \text{ and } x_j = 0,\\
\blue & \text{if } i = i_j \text{ and } x_j = 1,\\
\black & \text{if } i \neq i_j.\end{cases}
\]
Finally, the algorithm $\cB$ emulates the algorithm $\cA$ on input $y$, querying
the value of $x_j$ whenever $\cA$ queries the bit $(i_j,j)$ for some $j \le m$. This construction guarantees that $\cB$ computes $\gapID$ with error at most $\epsilon$; its query complexity corresponds to the number of $\red$ or $\blue$ entries that are queried by $\cA$.

Let $Q \subseteq [n] \times [m]$ be the random variable that denotes the set of indices queried by $\cA$, and let $C \subseteq [m]$ denote the set of columns whose $\red$ or $\blue$ entry is queried by $\cA$. Without loss of generality, we may assume that $\cA$ does not query any entry of a column after it finds the colored entry within that column. We partition $C$ into two sets $C_{\mathrm{early}}$ and $C_{\mathrm{late}}$, where $C_{\mathrm{early}}$ denotes the set of columns whose colored entry is found within the first $\frac n2$ queries to that column and $C_{\mathrm{late}}$ denotes the set of columns whose colored entry was found with more than $\frac n2$ queries to that column. Let $X_1,X_2,\ldots,X_{|Q|}$ be indicator variables where $X_k = 1$ if and only if the $k$th query $(i,j)$ made by $\cA$ is $\red$ or $\blue$ \emph{and} is one of the first $\frac n2$ queries to column $j$. Since each value $i_j$ is drawn uniformly at random from $[n]$, each of these indicator variables has expected value
$\E[X_k] \le \frac2n$. Therefore,
\[
\E\big[|C_{\mathrm{early}}|\big] = \E\left[ \sum_{i \le |Q|} X_i\right] \le \frac 2n \E\big[ |Q| \big].
\]
Furthermore, by definition at least $\frac n2$ queries are made to each column in $C_{\mathrm{late}}$ so the expected size of this set is bounded by
$\E\big[|C_{\mathrm{late}}|\big] \le \frac 2n \E \big[ |Q| \big]$ and
\[
\E\big[|C|\big] = \E\big[|C_{\mathrm{early}}|\big] + \E\big[|C_{\mathrm{late}}|\big] \le \frac4n \E\big[ |Q| \big].
\]
Thus, the expected query cost of $\cB$ is at most $\frac{4}{n} \cdot \overline{\rdt}_{\eps}(\bluered)$, as we wanted to show.
\end{proof}

\subsection{Lower Bound on the Query Complexity of PtrFcn}
\label{sec:ptrfcn-lb}

\begin{lemma}\label{lem:fnm-lb}
For every $0 \le \epsilon \le \frac14$,
$
\overline{\rdt}_{\epsilon}(\PtrFcn) \ge \overline{\rdt}_{2\epsilon}(\bluered).
$
\end{lemma}

\begin{proof}
Let $\cA$ be a randomized algorithm that computes $\PtrFcn$ with error at most
$\epsilon$ and expected query cost $q := \overline{\rdt}_{\epsilon}(\PtrFcn)$. We use $\cA$ to construct a randomized algorithm $\cB$ that
computes $\bluered$ with the same cost and error at most $2\epsilon$.

Let $x$ be an input for $\bluered$.
Each time $\cA$ queries a cell, $\cB$ queries the corresponding entry in $x$.
If the entry in $x$ is $\black$, then $\cB$ returns
$\langle1,\bot, \ldots,\bot \rangle$.
If the entry in $x$ is $\red$, then $\cB$ returns
$\langle0,\bot, \ldots, \bot\rangle$.
Finally, if the entry of $x$ is $\blue$, then $\cB$ terminates the emulation and returns $0$.
If $\cA$ reaches the end of the emulation without having been terminated, $\cB$ outputs the same result as $\cA$.

The query complexity of $\cB$ is at most that of $\cA$. It remains to
show that $\cB$ errs with probability at most $2\epsilon$. There are two
cases to consider.

The first case is when $x \in \bluered^{-1}(1)$.
Then $x$ maps directly to an input $z \in \PtrFcn^{-1}(0)$
and hence $\cB$ errs with probability at most $\epsilon$ on $x$.

The second case is when $x \in \bluered^{-1}(0)$.
Let $z$ be an arbitrary $1$-input for $\PtrFcn$ such that
(i) $z_{i,j} = \langle1,\bot,\ldots, \bot\rangle$ whenever $x_{i,j} = \black$,
(ii) $z_{i,j} = \langle0,\bot, \ldots, \bot\rangle$ whenever $x_{i,j} = \red$, and
(iii) the special entry and good entries of $z$ correspond to $\blue$ entries of $x$.
It might not be possible to completely emulate $\cA$ on input $z$ without knowing the exact set of $\blue$ entries.
However, $\cB$ doesn't need to fully emulate $\cA$---it only needs to know how to map $\black$ and $\red$ entries.
Once a $\blue$ entry is probed, $\cB$ halts and outputs $0$.
In this way, we claim that $\cB$ on input $x$ probes the same cells as $\cA$ on input $z$ until it halts.
Therefore its output is the same as $\cA(z)$ unless $\cA(z)$ probes the special cell or a good cell.
Moreover, in this case, $\cB$ outputs correctly with certainty.
Thus, by Claim~\ref{claim:specialcell}, the error of $\cB$ is at most
\[
\Pr[\cB \text{ errs}] \leq \Pr[\cB \text{ probes
      no blue cells}] \leq \Pr[\cA \text{ doesn't probe special cell}]
  \leq 2\eps\ . \qedhere
\]
\end{proof}

\subsection{Upper Bound on the Query Complexity of PtrFcn}
\label{sec:ptrfcn-ub}

The proof of Theorem~\ref{thm:separation} also requires a tight upper bound on the (worst-case) randomized query complexity of \PtrFcn. This argument is straightforward, and similar to the analysis of Ambainis et al.~\cite{AmbainisBBLSS17} for their analogous pointer function.

\begin{algorithm}[t]
\caption{PtrFcnSolver($x$)}
\label{alg:PtrFcn}
$S \gets [m]$\;
$T \gets$ a random subset of $[m]$ of size $|T| = \log \frac1\epsilon$\;
\For{each cell $(i,j)$ in a column in $T$}{
  \If{$\val(x_{i,j}) = 0 \wedge \col(x_{i,j}) \in S$}{
    $j^* \gets \col(x_{i,j})$\;
    $i^* \gets \row_{j^*}(x_{i,j})$\;
    $valid \gets \textsc{True}$\;
    \While{$|S| > 1 \wedge valid$}{
      $\ell \gets$ any column in $S \setminus \{j^*\}$\;
      \If{ $\val(x_{\row_\ell(x_{i^*,j^*}), \ell}) = 0$}{
        $S \gets S \setminus \{\ell\}$\;
      } \Else{
        $valid \gets \textsc{False}$\;
      }
    }
    \If {$|S| = 1$} { break;}
  }
}

\If{$|S| = 1$}{ fix $j \in S$.  \Return 1 if (i)  column $j$ is special, (ii)
  there is a special cell within column $j$, (iii) all cells linked by the
  special cell have value $0$, and (iv) half of linked cells point back to
  the special cell.

} \Return 0
\caption{PtrFcnSolver($x$)}
\end{algorithm}

\begin{lemma}
\label{lem:fnm-ub}
  $\rdt_{\eps}(\PtrFcn) = O( n\log \tfrac1\epsilon + m)$.
\end{lemma}

\begin{proof}
  The algorithm that computes the \PtrFcn function is described in
  Algorithm~\ref{alg:PtrFcn}. In this algorithm, the set $S$
  corresponds to the set of potential special columns.  The query
  complexity of PtrFcnSolver follows from the fact that each iteration
  of the inner while loop either eliminates one of the columns from
  the set $S$ of candidates or one of the $n \log \frac1\epsilon$
  cells in the columns in $T$. The final check of the (lone remaining)
  potential special column at the end of the algorithm examines at
  most $n + m$ cells.

  Whenever the PtrFcnSolver returns the value 1, then it in fact has
  observed a certificate that $\PtrFcn(x) = 1$ so the algorithm has
  perfect soundness.

  Conversely, suppose $\PtrFcn(x) = 1$.  Exactly half of the columns
  are good, so $T$ contains such a cell with probability at least
  $1-(1/2)^{\log(1/\eps)} = 1-\eps$.  Now, consider the for loop
  iteration when the first good cell $(i,j)$ is selected.  Since $(i,j)$ is a
  good cell, it points back to the special cell, which in turn points to a
  linked cell in all columns except the special column.  For any
  remaining $j'\neq j \in S$, $\PtrFcn(x)$ probes the linked cell in
  column $j'$, verifies the value equals $0$, and removes it from $S$.
  In this way, the remaining columns in $S$ save the special column
  are eliminated.  Once we reduce $S$ to a single remaining candidate,
  we can probe all cells in this column and all linked cells using
  $n+m$ queries to verify that indeed $\PtrFcn(x) = 1$.
\end{proof}

\subsection{Completing the Proof of Theorem~\ref{thm:separation}}
\label{sec:thmsep-proof}

The last ingredient that we need to complete the proof of Theorem~\ref{thm:separation} is the concept of \emph{resilient functions}~\cite{ChorGHFRS85}.

\begin{definition}
The function $\phi : \{0,1\}^n \to \{0,1\}^m$ is \emph{$t$-resilient}
for some $1 \le t < n$ if for any set $S \subseteq [n]$ of $|S| \le t$
coordinates and any assignment of values for the inputs
$\{x_i\}_{i \in S}$, when the values $\{x_i\}_{i \in [n] \setminus S}$
are set uniformly at random then $\phi(x)$ is uniformly distributed in
$\{0,1\}^m$.
\end{definition}

We use the following existence result on resilient functions that was
established by Chor et al.~\cite{ChorGHFRS85}.

\begin{theorem}[Chor et al.~\cite{ChorGHFRS85}]
  For every large enough $n$, there is a function
  $\phi : \{0,1\}^n \to \{0,1\}^m$ that is $\frac{n}{3}$-resilient and
  satisfies $m \ge 0.08n$.
\end{theorem}

We use resilient functions to bound the query complexity of functions
via the following lemma.

\begin{lemma}
\label{lem:resilient}
Fix a finite set $\cX$ of cardinality $|\cX| = 2^\ell$ for some
integer $\ell \ge 1$ and let $\phi : \{0,1\}^\N \to \cX$ be an
$\frac N3$-resilient function. Then for every function
$f : \cX^m \to \{0,1\}$ and every $\epsilon \ge 0$,
\[
\rdt_\epsilon(f \circ \phi) = \Theta( \N \cdot \rdt_\epsilon(f))
\qquad \mbox{and} \qquad
\overline{\rdt}_\epsilon(f \circ \phi) = \Theta( \N \cdot \overline{\rdt}_\epsilon(f)).
\]
\end{lemma}

\begin{proof}
  The upper bounds follow immediately from the observation that if
  $\cA$ is a randomized algorithm that computes $f$ with
  $\epsilon$-error, then we can define a algorithm $\cB$ that computes
  $f \circ \phi$ with the same error by simulating $\cA$ and querying
  the $\N$ bits to observe the value $\phi(x)$ to return to each query.

  For the lower bounds, let $\cA$ be a randomized algorithm that
  computes $f \circ \phi$ with error at most $\epsilon$. We define an
  algorithm $\cB$ for computing $f$ that simulates $\cA$ in the
  following way. For the first $\frac \N3$ queries to a cell, $\cB$
  answers the queries with uniformly random variables in $\{0,1\}$.
  On a query to the $(\frac \N3 + 1)$-th bit of a cell, $\cB$ queries the
  value $v$ of the corresponding cell in $x$.  It then draws a value
  $z$ in $\phi^{-1}(v)$ uniformly at random among all values that
  agree with the $\frac \N3$ bits output so far.  The current query and
  all further queries to bits of that cell are then answered using
  $z$.  Once $\cA$ terminates, $\cB$ returns $\cA$'s output and
  terminates as well.

  The correctness of $\cB$ follows directly from the correctness of
  $\cA$.  Furthermore, on any input for which $\cA$ makes $q$ queries,
  $\cB$ makes at most $q / (\N/3)$ queries since $\N/3$ distinct queries
  of $\cA$ are required for each query that $\cB$ eventually makes to
  $x$. Thus both the average-case and worst-case query complexities of
  $\cB$ are bounded by $3/\N$ times the corresponding query
  complexities of $\cA$.
\end{proof}

We are now ready to complete the proof of the separation theorem.

\begin{proof}[Proof of Theorem~\ref{thm:separation}]
Fix $m = n = 2^\ell - 1$ for any integer $\ell \ge 1$ so that $|\Gamma| = 2^{\ell(2^\ell-1) + \ell + 1}$ is a power of 2.
Fix a $\frac{C}3$-resilient function $\phi : \{0,1\}^C \to \Gamma$ for some $C \le 12.5\log |\Gamma|$ and define the function $\EncFcn = \PtrFcn \circ \phi$. By Lemmas~\ref{lem:resilient} and~\ref{lem:fnm-ub},
\[
\rdt_{\eps}(\EncFcn) = O\big(C(n\log \tfrac1\epsilon+m)\big)
= O\big(Cn\log \tfrac1\epsilon\big).
\]
In particular, setting $\epsilon = \frac13$ we obtain
$\rdt(\EncFcn) = O(C n)$.

Using Lemma~\ref{lem:resilient},~\ref{lem:fnm-lb}, and~\ref{lem:bluered-lb},
we obtain the chain of inequalities
\[
\overline{\rdt}_{\eps}(\EncFcn)
= \Omega\big( C \cdot \overline{\rdt}_{\epsilon}(\PtrFcn) \big)
= \Omega\big( C \cdot \overline{\rdt}_{2\epsilon}(\bluered) \big)
= \Omega\big( Cn \cdot \overline{\rdt}_{2\eps}(\gapID) \big).
\]
By Lemma~\ref{lem:gapid-lb}, when $\epsilon > 2^{-m} = 2^{-n}$ this implies that
\[
\overline{\rdt}_{\eps}(\EncFcn)
= \Omega\big( Cn \log \tfrac1\epsilon\big)
= \Omega\big( \log \tfrac1\epsilon \cdot \rdt(\EncFcn) \big).
\]
Theorem~\ref{thm:separation} is obtained by noting that $\EncFcn$ is a function on $N = O( m n |\Gamma| ) = O(n^3 \log n)$ variables.
\end{proof}

\section{Strong Direct Sum Theorem}
\label{sec:ds}

We establish Theorem~\ref{thm:directsum} by proving a corresponding direct
sum result in the distributional model and applying a Yao minimax principle
for algorithms that err and abort with bounded probability.

We introduce the model of algorithms that can abort in Section~\ref{sec:abort}, where we also relate this model to the average query complexity setting of randomized algorithms and establish a Yao minimax principle.
In Section~\ref{sec:directsum-dist}, we establish the main technical result, a strong direct sum theorem for distributional complexity. We complete the proof of Theorem~\ref{thm:directsum} itself in Section~\ref{sec:directsum-thm} and the proofs of Corollaries~\ref{cor:directsum-rdt} and~\ref{cor:directsum-rcc} are completed in Section~\ref{sec:directsum-cors}.

\subsection{Algorithms That Can Abort}
\label{sec:abort}

We consider randomized algorithms that are allowed to \emph{err} and \emph{abort}.
In this setting, an algorithm outputs $\bot$ instead of giving a valid output when it chooses to abort.
Let $\ddt_{\delta,\eps}^\mu(f)$ and $\rdtde(f)$ denote the distributional
and randomized query complexities of $f$ when the algorithm aborts
with probablity at most $\delta$ and errs with probability at most
$\eps$.

Randomized query complexity in the setting where algorithms can abort with constant probability $\delta$ is asymptotically equivalent to the average randomized query complexity.

\begin{proposition}
\label{prop:abort-average}
For every function $f : \bitsn \to \bit$, every $0 \le \epsilon < \frac12$ and every $0 < \delta < 1$,
\[
\delta \cdot \rdtde(f) \le \overline{\rdt}_\epsilon(f) \le \tfrac{1}{1-\delta} \cdot \rdt_{\delta,(1-\delta)\epsilon}(f).
\]
\end{proposition}

\begin{proof}
  For the first inequality, let $\cA$ be a randomized algorithm that
  computes $f$ with $\epsilon$ error and has expected query complexity
  $q$. Let $\cB$ be the randomized algorithm $\cB$ that simulates
  $\cA$ except that whenever $\cA$ tries to make more than $q/\delta$
  queries, it aborts. The algorithm $\cB$ also computes $f$ with
  error at most $\epsilon$, and it has worst-case query complexity
  $q/\delta$. Furthermore, by Markov's inequality, $\cB$ aborts with
  probability at most $\delta$.

  For the second inequality, let $\cB$ be a randomized algorithm with
  query complexity $q$ that computes $f$ with error probability at
  most $(1-\delta)\epsilon$ and abort probability at most
  $\delta$. Let $\cA$ be the randomized algorithm that simulates
  $\cB$ until that algorithm does not abort, then outputs the same
  value. The error probability of $\cB$ conditioned on it not
  aborting is at most
  $\frac{(1-\delta)\epsilon}{1-\delta} = \epsilon$, so the algorithm
  $\cA$ also errs with probability at most $\epsilon$, and its
  expected query complexity is
  $q(1 + \delta + \delta^2 + \cdots) = \frac{q}{1-\delta}$.
\end{proof}

Yao's minimax principle can be adapted for the setting of algorithms that abort as follows.

\begin{lemma}
\label{lem:yao-abort}
For any $\alpha,\beta>0$ such that $\alpha+\beta \leq 1$, we have
\[
\max_\mu \ddt_{\delta/\alpha, \eps/\beta}^\mu(f) \leq \rdtde(f)
    \leq \max_{\mu} \ddt_{\alpha\delta, \beta\eps}^\mu(f).
\]
\end{lemma}

\begin{proof}
    We handle the initial inequality (i.e., the \emph{easy direction})
    first.  Fix a $q$-query randomized algorithm $\cA$ achieving
    $\rdtde(f)$.  By the guarantee of $\cA$, we have that for any
    input $x$, $\cA$ aborts with probability at most $\delta$ and errs
    with probabiltiy at most $\eps$.  Let $\ind_\delta(x)$ and
    $\ind_\eps(x)$ be indicator variables for the events that $\cA$
    aborts on $x$ and $\cA$ errs on $x$ respectively.  Then, we have
    $\E_R[\ind_\delta(x)] \leq \delta$ and similarly
    $\E_R[\ind_\eps(x)] \leq \eps$ when the expectation is taken over the randomness $R$ of the algorithm $\cA$.  Next, fix any input distribution
    $\mu$ and let $X \sim \mu$.  It follows that
\[
\E_R\left[\E_X[\ind_\delta(X)]\right] =
    \E_X\left[\E_R[\ind_\delta(X)]\right] \leq \delta \quad \text{and}
    \quad \E_R\left[\E_X[\ind_\eps(X)]\right] =
    \E_X\left[\E_R[\ind_\eps(X)]\right] \leq \eps.
\]
Using Markov's inequality twice, we have
\[
\Pr_R\left[\E_X[\ind_\delta(X)] >
      \delta/\alpha\right] < \alpha \qquad \text{and} \qquad
    \Pr_R\left[\E_X[\ind_\eps(X)] > \eps/\beta \right] < \beta.
\]
By a union bound, there exists a setting of the random string $R$
    such that both $\E[\ind_\delta(X)] \leq \delta/\alpha$ and
    $\E[\ind_\eps(X)] \leq \eps/\beta$.
Fixing this $R$ gives a
    $q$-query deterministic algorithm that aborts with probability at
    most $\delta/\alpha$ and errs with probability at most
    $\eps/\beta$, hence
    $\ddt_{\delta/\alpha, \eps/\beta}^\mu(f) \leq \rdt_{\delta,\eps}(f)$.

    For the second inequality, let $c \deq \max_\mu \ddt_{\alpha\delta,
      \beta\eps}^\mu(f)$.  Consider a two-player, zero-sum game where
    player 1 selects a $c$-query deterministic algortihm $\cA$ for $f$,
    player 2 selects an input $x$, and player 1 is paid $-\eps$ if
    $\cA(x)$ aborts, $-\delta$ if $\cA(x)$ errs, and $0$ otherwise.
    Note that each mixed strategy for player 1 corresponds to a
    randomized algorithm and each mixed strategy for player 2
    corresponds to an input distribution $\mu$.  By our choice of $c$,
    it follows that for any mixed strategy for player 2, player 1 can
    obtain payoff $-\eps(\alpha\delta) -\delta(\beta\eps) \geq
    -\eps\delta$.  By the minimax theorem, it follows that there is a
    mixed strategy for player 1 (i.e., a $c$-query randomized
    algorithm $\cA$) that provides the same payoff for every choice of
    player 2.  Finally, note that $\cA$ aborts with probability at
    most $\delta$ and errs with probability at most $\eps$; otherwise,
    the payoff would be less than $-\eps\delta \leq
    -\eps\delta(\alpha+\beta)$.  We've shown a $c$-query randomized
    algorithm that aborts w/probability at most $\delta$ and errs
    w/probability at most $\eps$, hence $\rdtde(f) \leq c = \max_\mu
    \ddt_{\alpha\delta, \beta\eps}^\mu(f)$.
\end{proof}

\subsection{Strong Direct Sum for Distributional Complexity}
\label{sec:directsum-dist}

We prove a slightly more precise variant of Lemma~\ref{lem:main-ds}.

\begin{lemma}
\label{lem:ds-Dmu2}
For every function $f : \{0,1\}^n \to \{0,1\}$, every distribution $\mu$ on $\{0,1\}^n$, and every $0 \le \delta, \epsilon \le \frac14$,
\[
\ddt^{\mu^k}_{\delta,\epsilon}(f^k) = \Omega\left( k \cdot \ddt^\mu_{\frac1{10}+4\delta+4\epsilon,\frac{48\epsilon}{k}}(f) \right).
\]
\end{lemma}

\begin{proof}
Let $\cA$ be a deterministic algorithm with query complexity $q$ that computes $f^k$ with error probability at most $\epsilon$ and abort probability at most $\delta$ when the input $x = (x^{(1)},\ldots,x^{(k)})$ is drawn from $\mu^k$. Then conditioned on $\cA$ not aborting, it outputs the correct value of $f^k$ with probability at least $1-\frac{\epsilon}{1-\delta} \ge 1 - 2\epsilon$ and
\begin{align*}
1-2\epsilon
&\le \Pr_{x \sim \mu^k}\left[ \cA(x) = f^k(x) \;\middle|\; \cA(x) \neq \bot \right] \\
&= \prod_{i \le k} \Pr_{x \sim \mu^k}\left[ \cA(x)_i = f(x^{(i)}) \;\middle|\; \cA(x)_{< i} = f^k(x)_{< i}, \cA(x) \neq \bot \right].
\end{align*}
This implies that at least $\frac23k$ indices $i \in [k]$ satisfy
\begin{equation}
\label{eqn:error-cond}
\Pr_{x \sim \mu^k}\left[ \cA(x)_i \neq f(x^{(i)}) \;\middle|\; \cA(x)_{< i} =  f^k(x)_{< i}, \cA(x) \neq \bot \right] \le \frac{12\epsilon}{k},
\end{equation}
otherwise the product in the product in the previous inequality would be less than $(1-12\epsilon/k)^{k/3} \le e^{-4 \epsilon} < 1-2\epsilon$, contradicting the lower bound on this product.

For each $i \le k$, let $q_i(x)$ denote the number of queries that $\cA$ makes to $x^{(i)}$ on input $x$. The query complexity of $\cA$ guarantees that for each input $x$, $\sum_{i \le k} q_i(x) \le q$. Therefore, $\sum_{i \le k} \E_{x \sim \mu^k} [ q_i(x) ] \le q$ and at least $\frac23k$ indices $i \in [k]$ satisfy
\begin{equation}
\label{eqn:query-cond}
\E_{x \sim \mu^k}\left[ q_i(x) \right] \le \frac{3q}{k}.
\end{equation}
Thus, some index $i^* \in [k]$ satisfies both~\eqref{eqn:error-cond} and~\eqref{eqn:query-cond}. Fix such an index $i^*$. For inputs $y \in \mu^k$ and $x \in \mu$, write $y^{(i^* \gets x)} := (y^{(1)},\ldots,y^{(i^*-1)},x,y^{(i^*+1)},\ldots,y^{(k)})$ to be the input obtained by replacing $y^{(i^*)}$ with $x$ in $y$. With this notation, the two conditions~\eqref{eqn:error-cond} and~\eqref{eqn:query-cond} satisfied by $i^*$ can be rewritten as
\[
\E_{y \sim \mu^k}\left[ \Pr_{x \sim \mu}\left[ \cA(y^{(i^* \gets x)})_{i^*} \neq f(x) \;\middle|\; \cA(y^{(i^* \gets x)})_{< i^*} = f^k(y^{(i^* \gets x)})_{< i^*}, \cA(y^{(i^* \gets x)}) \neq \bot \right]\right] \le \frac{12\epsilon}{k}
\]
and
\[
\E_{y \sim \mu^{k}}\left[ \E_{x \sim \mu}\left[ q_{i^*}(y^{(i^* \gets x)}) \right]\right] \le \frac{3q}{k}.
\]
The correctness of $\cA$ also guarantees that
\[
\E_{y \sim \mu^{k}}\left[ \Pr_{x \sim \mu}\left[ \cA(y^{(i^* \gets x)}) = \bot \right]\right] \le \delta
\]
and
\[
\E_{y \sim \mu^{k}}\left[ \Pr_{x \sim \mu}\left[ \cA(y^{(i^* \gets x)})_{< i^*} \neq f^k(y^{(i^* \gets x)})_{< i^*}) \;\middle|\; \cA(y^{(i^* \gets x)}) \neq \bot \right]\right] \le \epsilon.
\]
Therefore, by Markov's inequality, there exists an input $z \in \{0,1\}^{n \times k}$ such that
\begin{align*}
\Pr_{x \sim \mu}\left[ \cA(z^{(i^* \gets x)}) = \bot \right] &\le 4\delta, \\
\Pr_{x \sim \mu}\left[ \cA(z^{(i^* \gets x)})_{< i^*} \neq f^k(z^{(i^* \gets x)})_{< i^*} \;\middle|\; \cA(z^{(i^* \gets x)}) \neq \bot \right] &\le 4\epsilon, \\
\Pr_{x \sim \mu}\left[ \cA(z^{(i^* \gets x)})_{i^*} \neq f(x) \;\middle|\; \cA(z^{(i^* \gets x)})_{< i^*} = f^k(z^{(i^* \gets x)}), \cA(z^{(i^* \gets x)})\neq \bot \right] &\le \frac{48\epsilon}{k}, \mbox{ and}\\
\E_{x \sim \mu}\left[ q_{i^*}(z^{(i^* \gets x)}) \right] &\le \frac{12q}{k}.
\end{align*}

Let $\cA'$ be the deterministic algorithm that computes $f(x)$ by
simulating $\cA$ on the input $z^{(i^* \gets x)}$ with two
additions:
\begin{enumerate}
\item If $\cA$ attempts to query more than $\frac{120q}{k}$ bits of $x$, $\cA'$ aborts, and
\item When $\cA$ terminates, the algorithm $\cA'$ first verifies that the output generated by $\cA$ satisfies $\cA(z^{(i^* \gets x)})_{\le i^*} = f^k(z^{(i^* \gets x)})$. If so $\cA'$ returns the value $\cA(z^{(i^* \gets x)})_{i^*}$; if not, $\cA'$ aborts.
\end{enumerate}
The algorithm $\cA'$ has query complexity at most $\frac{120q}{k}$ and, by the conditions satisfied by $z$, it aborts with probability at most
$\frac1{10} + 4\delta + 4\epsilon$ and errs
with probability at most $\frac{48\epsilon}{k}$ when $x \sim \mu$.
\end{proof}

\subsection{Proof of Theorem~\ref{thm:directsum}}
\label{sec:directsum-thm}

We now complete the proof of Theorem~\ref{thm:directsum}.
Fix $\delta = \frac1{40}$.
By Proposition~\ref{prop:abort-average} and the second inequality of Lemma~\ref{lem:yao-abort},
\[
\overline{\rdt}_{\frac{96\epsilon}k}(f)
\le 2\,\rdt_{\frac12,\frac{48\epsilon}k}(f) \le
2\,\rdt_{\frac1{5}+4\delta+4\epsilon,\frac{48\epsilon}{k}}(f)
\le 2\max_\mu \ddt^\mu_{\frac1{10}+2\delta+2\epsilon,\frac{24\epsilon}{k}}(f).
\]
Let $\mu^*$ denote a distribution where the maximum is attained. By Lemma~\ref{lem:ds-Dmu2},
\[
\ddt^{\mu^*}_{\frac1{10}+2\delta+2\epsilon,\frac{24\epsilon}{k}}(f)
= O\left( \frac1k \cdot \ddt^{(\mu^*)^k}_{\frac{\delta}2,\frac\epsilon2}(f^k)\right).
\]
Using the first inequality of Lemma~\ref{lem:yao-abort} we then obtain
\[
\ddt^{(\mu^*)^k}_{\frac{\delta}2,\frac\epsilon2}(f^k)
\le \max_{\nu} \ddt^{\nu}_{\frac\delta2,\frac\epsilon2}(f^k)
\le \rdt_{\delta,\epsilon}(f^k).
\]
Combining these inequalities and applying Proposition~\ref{prop:abort-average} once more yields
\[
\overline{\rdt}_{\frac{96\epsilon}k}(f)
\le O\big( \tfrac1k \cdot \rdt_{\delta,\epsilon}(f^k) \big)
\le O\big( \tfrac1k \cdot \overline{\rdt}_\epsilon(f^k) \big).
\]
Theorem~\ref{thm:directsum} follows from the identity
$\overline{\rdt}_{\frac{\epsilon}k}(f) = \Theta\big(\overline{\rdt}_{\frac{96\epsilon}k}(f)\big)$ obtained from the standard success amplification trick.
\qed

\subsection{Proof of Corollaries~\ref{cor:directsum-rdt} and~\ref{cor:directsum-rcc}}
\label{sec:directsum-cors}

Corollary~\ref{cor:directsum-rdt} is obtained as a direct consequence of Theorems~\ref{thm:separation} and~\ref{thm:directsum}.

\begin{proof}[Proof of Corollary~\ref{cor:directsum-rdt}]
The upper bound is via the universal bound~\eqref{eq:ds-ub}. For the matching lower bound,
let $f : \bitsn \to \bit$ be a function that satisfies the condition of Theorem~\ref{thm:separation}. By Theorem~\ref{thm:directsum}, the randomized communication complexity of $f^k$ satisfies
\[
R(f^k) \ge \overline{\rdt}(f^k) = \Omega\big( k \cdot \overline{\rdt}_{\frac1{3k}}(f)\big)
\]
By Theorem~\ref{thm:separation},
\[
\overline{\rdt}_{\frac1{3k}}(f) = \Omega\big(\rdt(f) \cdot \log k \big)
\]
as long as $k \le 2^{\errordep}$.
Combining those inequalities yields $\rdt(f^k) = \Omega\big(k \log k \cdot \rdt(f)\big)$,
as we wanted to show.
\end{proof}

The proof of Corollary~\ref{cor:directsum-rcc} uses the following randomized query-to-communication lifting theorem of \goos, Pitassi, and Watson~\cite{GoosPW17}.

\begin{theorem}[\goos, Pitassi, Watson]
\label{thm:gpw-lift}
Define $\textsc{Ind}_m : [m] \times \{0,1\}^m \to \{0,1\}$ to be the index function mapping $(x,y)$ to $y_x$ and fix $m = n^{256}$.
For every $f : \bitsn \to \bit$,
\[
\rcc(f \circ \textsc{Ind}_m) = \rdt(f) \cdot \Theta( \log n)
\]
and
\[
\rcc(f^k \circ \textsc{Ind}_m) = \rdt(f^k) \cdot \Theta( \log n).
\]
\end{theorem}

\begin{remark}
The statement of Theorem~\ref{thm:gpw-lift} in~\cite{GoosPW17} only mentions the first identity explicitly. However, as discussed in their Section II, the theorem statement holds for functions with any finite range.\footnote{In fact, their theorem also holds in even more general settings such as when $f$ is a partial function or a relation, for example.} Therefore, the theorem holds for the function $f^k$ as well as $f$.
\end{remark}

\begin{proof}[Proof of Corollary~\ref{cor:directsum-rcc}]
By Corollary~\ref{cor:directsum-rdt}, there exists a function $f : \bitsn \to \bit$ which satisfies $\rdt(f^k) = \Theta(k \log k \cdot \rdt(f))$. Combining this result with Theorem~\ref{thm:gpw-lift}, we obtain
\begin{align*}
\rcc\big( (f \circ \textsc{Ind}_m)^k \big)
&= \rcc\big( f^k \circ \textsc{Ind}_m \big) \\
&= \rdt(f^k) \cdot \Theta(\log n) \\
&= \Theta( k \log k \cdot \rdt(f) \cdot \log n) \\
&= \Theta( k \log k ) \cdot \rcc(f\circ \textsc{Ind}_m). \qedhere
\end{align*}
\end{proof}

\section{Conclusions and Open Problems}

In this work, we prove a strong direct sum theorem for average-case
randomized query complexity---to compute $f^k$ with probability
$\varepsilon$ in the average-case query complexity model, one must
spend $k$ times the resources of computing $f$ with error
$\varepsilon/k$.  We then give the first total function $f$ whose
query complexity scales with $\log(1/\varepsilon)$, matching the
blowup one gets from standard error reduction.
We believe a strong direct sum-type theorem should also hold for many
\emph{composed functions}.  A natural first step would be to prove an \emph{XOR Lemma} in the strong direct sum setting.

\begin{conjecture}{(Strong Direct Sum for $\xor$.)}\label{conj:xor}
For all functions $f$, all positive integers $k \geq 2$, and all $0 < \varepsilon < 1/3$,
$$\overline{\rdt}_\eps(\xor_k\circ f) = \Theta(k\cdot \overline{\rdt}_{\varepsilon/k}(f))\ .$$
\end{conjecture}

We believe a similar result should hold for the majority function $\maj_k$.
\begin{conjecture}{(Strong Direct Sum for $\maj$.)}\label{conj:maj}
  For all functions $f$, all positive integers $k \geq 2$, and all $0 < \varepsilon < 1/3$,
  $$\overline{\rdt}_\eps(\maj_k\circ f) = \Theta(k \cdot \overline{\rdt}_{\varepsilon/k}(f))\ .$$
\end{conjecture}

More generally, let us say that a total function $g:\bit^k \rightarrow\bit$ \emph{admits a strong direct sum theorem} if for all functions $f:\bitsn \rightarrow \bit$, for all positive integers $k\geq 2$, and for all $0 < \eps < 1/3$, we have $\overline{\rdt}_\eps(g\circ f) = \Theta(k\cdot \overline{\rdt}_{\eps/k}(f))$.
Using this terminology, Conjectures~\ref{conj:xor} and~\ref{conj:maj} posits that $\xor_k$ and $\maj_k$ admit strong direct sum theorems.
These two conjectures are special cases of the following general problem.

\begin{oproblem}
  Which functions $g:\bit^k \rightarrow\bit$ admit strong direct sum theorems?
\end{oproblem}

In a different direction, we note that
while we identified one total function whose randomized query complexity scales logarithmically with the inverse of the error parameter $\varepsilon$, we don't have a good understanding of which functions scale this way.

\begin{oproblem}
Characterize the functions $f:\bitsn \rightarrow \bit$ whose average-case randomized query complexity satisfies $\overline{R}_{\eps}(f) = \Omega(R_{1/3}(f)\cdot \log(1/\eps))$ for all $\eps > 2^{-n^{\Omega(1)}}$.
\end{oproblem}

\section*{Acknowledgments}

The first author thanks Alexander Belov and Shalev Ben-David for enlightening discussions and helpful suggestions.
The second author~thanks Peter Winkler for insightful discussions.
Both authors wish to thank the anonymous referees for valuable feedback and for the reference to~\cite{Sherstov18}.

\bibliographystyle{plain}
\bibliography{RandQueryComplexity}

\end{document}